\documentclass[english]{lipics-v2021}

\usepackage{graphicx}
\graphicspath{{./figure/}}
\newcommand{\etal}{\emph{et~al.}\xspace}

\title{Linear time single-source shortest path algorithms in Euclidean graph classes}

\author{Joachim Gudmundsson}{The University of Sydney, Australia}
{joachim.gudmundsson@sydney.edu.au}{}{}
\author{Yuan Sha}{The University of Sydney, Australia}
{hubeishayuan@gmail.com}
{}{}
\author{Sampson Wong}{The University of Copenhagen, Denmark}
{sawo@di.ku.dk}{}{}

\authorrunning{J. Gudmundsson, Y. Sha and S. Wong} 
\Copyright{Joachim Gudmundsson, Yuan Sha and Sampson Wong} 

\usepackage[T1]{fontenc}
\usepackage{xspace}

\newcommand{\R}{\mathbb{R}}

\newtheorem{fact}[theorem]{Fact}
\newtheorem{openproblem}[theorem]{Open Problem}

\nolinenumbers

\begin{document}
\maketitle
\begin{abstract}
In the celebrated paper of Henzinger, Klein, Rao and Subramanian (1997), it was shown that planar graphs admit a linear time single-source shortest path algorithm. Their algorithm unfortunately does not extend to Euclidean graph classes. We give criteria and prove that any Euclidean graph class satisfying the criteria admits a linear time single-source shortest path algorithm. As a main ingredient, we show that the contracted graphs of these Euclidean graph classes admit sublinear separators.

\end{abstract}

\section{Introduction}
Computing shortest paths is a fundamental problem in graph theory and in network analysis. Given a graph with non-negative edge weights, the single-source shortest path (SSSP) problem is to compute the shortest paths from a source vertex to all the other vertices in the graph. The SSSP problem has found numerous applications in transportation networks, geographic information systems, motion planning, social networks, chip design, and so on. 

The most well-known SSSP algorithm is Dijkstra's algorithm~\cite{dijkstra1959note}. Using a Fibonacci heap~\cite{fibonacci-heap_FredmanTarjan1984}, Dijkstra's algorithm runs in $O(m+n\log n)$ time on graphs with $m$ edges and $n$ vertices. Dijkstra's algorithm has a lower bound of $\Omega(n\log n)$ in the comparison computer model, since the algorithm sorts the vertices of the graph by their distances from the source vertex. 

In planar graphs, faster algorithms are known. Frederickson~\cite{shortest-path-planar_Frederickson1987} gave an $O(n\sqrt{\log n})$ time SSSP algorithm for planar graphs, using \emph{$r$-division} and graph contraction. An $r$-division of a planar graph divides the graph into $O(n/r)$ regions such that each region has $O(r)$ vertices and $O(\sqrt{r})$ boundary vertices. Henzinger, Klein, Rao and Subramanian~\cite{SSSP-planar_HenzingerKRS1997} use \emph{recursive division}, a recursive version of $r$-division, and an edge-relaxation algorithm on the recursive division to give an optimal linear-time SSSP algorithm for planar graphs. 
The algorithms leverage two key properties of planar graphs: (i) the graph class is minor-closed~\cite{kuratowski1930probleme}, and (ii) the graphs admit a sublinear separator that can be computed in linear time~\cite{Planar-separator-theorem_LiptonTarjan79}. 
Tazari and M{\"{u}}ller{-}Hannemann~\cite{DBLP:conf/wg/TazariM08} showed how to use knitted H-partition to avoid introducing arbitrarily large minors and extend the linear time algorithm of HKRS~\cite{SSSP-planar_HenzingerKRS1997} to any minor-closed graph class. A natural open problem is to ask: which other graph classes admit a linear time SSSP algorithm? 
\begin{openproblem}
    \label{openproblem1}
     Other than planar graphs and minor-closed graphs, which graph classes admit a linear time single-source shortest path algorithm? 
\end{openproblem}

Practical applications have motivated researchers to study Open~Problem~\ref{openproblem1}. For example, road networks are non-planar, and it is unknown if road networks are minor-closed, yet computing shortest paths remains a fundamental problem on road networks. 

Le and Than~\cite{lanky-graph-sublinear-separator_LeThan-soda2022} proposed $\tau$-lanky graphs, which are any Euclidean graph where any ball in $\mathbb R^d$ with radius~$r$ cuts at most $\tau$ edges longer than~$r$. Lanky graphs unify several Euclidean graph classes including greedy spanners in $\R^d$ (any fixed $d$), low-density graphs in $\R^d$ (\cite{map-matching-low-density_ChenDGNW-ALENEX2011}) and greedy spanners of unit ball graphs in $\R^d$. All these graph classes are $O(1)$-lanky graphs. The low-density graph in $\R^2$ is argued to be a realistic model for road networks~\cite{map-matching-low-density_ChenDGNW-ALENEX2011}.

In the conference version of their paper, Le and Than~\cite{lanky-graph-sublinear-separator_LeThan-soda2022} proved that lanky graphs in $\R^d$ admit sublinear separators. They claimed that their (expected) linear time sublinear separator algorithm for lanky graphs implies a linear time SSSP algorithm by applying HKRS's algorithm~\cite{SSSP-planar_HenzingerKRS1997}. Unfortunately, this claim has a gap. While lanky graphs admit sublinear separators, they are not minor-closed. Therefore, HKRS's algorithm does not directly extend to lanky graphs. The SSSP claim is retracted in the full version of the paper~\cite{DBLP:journals/corr/abs-2107-06490}. Prior to their paper, Eppstein and Khodabandeh~\cite{edge-crossing-greedy-spanner_EppsteinK2021} studied the edge crossing patterns of greedy spanners in $\R^2$. They gave an $O(n\log ^{(i)} n)$ time SSSP algorithm for greedy spanners in $\R^2$ (Corollary 21), where $i$ is any constant and $\log^{(i)} n$
denotes the $i$-times iterated
logarithm. These efforts raise the following open problem.

\begin{openproblem}
    \label{openproblem2}
    Is there a linear time SSSP algorithm for $\tau$-lanky graphs in $\R^d$?
\end{openproblem}

Miller, Teng, Thurston and Vavasis~\cite{geometric-approach-graph-separator_focs91,separator-sphere-packing_MillerEtal-jacm1997, separator-finite-element-meshes_MillerEtal1998} developed a geometric characterization of graphs that have a sublinear separator, using the notion of a \emph{$k$-ply neighborhood system}. A $k$-ply neighborhood system is a collection of balls in $\mathbb R^d$ such that no point in the space is covered by more than $k$ of the balls. They showed that (the intersection graphs of) $k$-ply neighborhood systems can model sphere-packings~\cite{separator-sphere-packing_MillerEtal-jacm1997}, $k$-nearest neighbor graphs~\cite{separator-sphere-packing_MillerEtal-jacm1997} and finite element meshes~\cite{DBLP:journals/siamsc/MillerTTV98}. They proved that $k$-ply neighborhood systems in $\R^d$ (and their intersection graphs) admit a sublinear separator that can be computed in linear time~\cite{separator-sphere-packing_MillerEtal-jacm1997,separator-k-ply-neighborhood_EppsteinMT-SCG93}. 

Eppstein and Goodrich~\cite{road-network-algorithmic-lens_EppsteinGoodrich-GIS2008} use $k$-ply neighborhood system in $\R^2$ as a model of real-world road networks. They showed how to compute SSSP in randomized linear time, assuming the arrangement of the $k$-ply disk neighborhood system is given. Besides the assumption, their approach is specific to $\R^2$ and does not extend to higher dimensions. This raises the following open problem.

\begin{openproblem}
    \label{openproblem3}
    Is there a linear time SSSP algorithm for the intersection graph of a $k$-ply neighborhood system in $\mathbb R^d$, where $d \geq 2$?
\end{openproblem}

Smith and Wormald~\cite{geometric-separator-theorems-focs98} considered a collection of cubes in $\R^d$ such that no point in the space is covered by more than $\kappa$ of the cubes. Here we call the collection of cubes a \emph{$\kappa$-thick cubical neighborhood system}. The authors proved that $\kappa$-thick cubical neighborhood systems in $\R^d$ (and their intersection graphs) admit a sublinear separator that can be computed in randomized linear time. We consider the following open problem.
\begin{openproblem}
    \label{openproblem4}
    Is there a linear time SSSP algorithm for the intersection graph of a $\kappa$-thick cubical neighborhood system in $\mathbb R^d$?
\end{openproblem}

\subsection{Our contribution}

Our main technical contribution is to compute a recursive division in linear time, for any Euclidean graph class that satisfies the following \textbf{criteria}:

\begin{enumerate}[(I)]
    \item The graph $G$ admits a sublinear separator that can be computed in linear time. \label{crit1}
    \item The combinatorially contracted graphs $G_i$ of $G$ (refer to Definition~\ref{def:contracted-graph}) are sparse, i.e., any $\bar{k}$-vertex subgraph of $G_i$ is sparse.
    
    \item A closed surface in $\R^d$ (such as a Jordan curve in $\R^2$, a sphere in $\R^d$) is used to find the separator in (\ref{crit1}). The closed surface cuts a sublinear number of geometrical objects (such as segments, balls) associated with the vertices in $G$, and these associated vertices form the separator. Moreover, a random closed surface cuts a sublinear number of geometrical objects associated with the vertices in $G$ in expectation. 
    \item The graph class of $G$ is subgraph-closed.
\end{enumerate}

Compared with HKRS~\cite{SSSP-planar_HenzingerKRS1997}, criteria (II)-(IV) replace the minor-closed condition, and can be satisfied by a family of non-minor-closed graph classes. Criteria (I)-(IV) are formally stated in Section~\ref{sec:main-theorem}, in which we use criteria (I)-(IV) to compute a recursive division of the Euclidean graph class in linear time. Once we have the recursive division\footnote{To be precise, the recursive division should satisfy the conditions in Inequalities~\eqref{eqn:tricky-condition}.}, we can apply HKRS's edge-relaxation algorithm on it to compute SSSP in linear time.

In particular, we show that $\tau$-lanky graphs, the intersection graphs of $k$-ply neighborhood systems and the intersection graphs of $\kappa$-thick cubical neighborhood systems all satisfy the above criteria. Thus we can compute a recursive division (and SSSP) for these Euclidean graph classes in linear time. In this way, we answer Open Problems~\ref{openproblem2}, ~\ref{openproblem3} and~\ref{openproblem4} affirmatively.

Our approach to constructing a recursive division in linear time uses contracted graphs of the Euclidean graph classes. The main technical challenges of using contracted graphs for Euclidean graph classes are:   
\begin{itemize}
    \item Edge contraction for planar or minor-closed graphs is a combinatorial operation. The contracted graphs in HKRS~\cite{SSSP-planar_HenzingerKRS1997} are obtained by such combinatorial edge contractions. In HKRS's approach, the resulting contracted graphs must remain in the original graph class to guarantee that they have sublinear separators. This is not the case for Euclidean graph classes.
    \item In Euclidean graph classes, vertices are points in~$\mathbb R^d$ and edges are segments in~$\mathbb R^d$. Unfortunately, edge contraction \emph{in the Euclidean space} does not preserve geometric properties and seems meaningless. For example, if an edge~$(a,b)\in \mathbb R^d$ is contracted into a new vertex~$c \in \mathbb R^d$, any edge $e_a$ or $e_b$ attached to $a$ or $b$ would now be attached to $c$. There is no way to choose $c$ so that the geometric properties of $e_a$ and $e_b$ (e.g. length, intersections) are preserved.
\end{itemize}

Although edge contractions in the Euclidean space are incompatible with the geometric properties of straight line embeddings, we overcome this obstacle by performing edge contractions combinatorially and proving that 1) the combinatorially contracted graphs admit sublinear balanced separators, and 2) such balanced separators can be computed in time linear to the size of the contracted graph. These properties are unknown before and we believe that they are of independent interest.

\subsection{Techniques}\label{ssec:tech-review}

A recursive division of a graph divides the graph into regions, then divides the regions into subregions, and so on recursively. There are two approaches to computing a recursive division of planar graph in linear time. The first approach is to adapt Goodrich's~\cite{planar-division-parallel-triangulation_Goodrich-stoc92} recursive separator decomposition algorithm to compute a recursive division of the planar graph in linear time. Goodrich's recursive separator decomposition algorithm emulates Lipton and Tarjan's planar separator algorithm, maintains tree structures used in LT's algorithm and other tree structures dynamically for the divided pieces during decomposition. In this way a separator of the divided piece can be computed in time sublinear to the size of the piece. The second approach, attributed to Frederickson~\cite{shortest-path-planar_Frederickson1987} and HKRS~\cite{SSSP-planar_HenzingerKRS1997}, is to contract the graph consecutively into a sequence of contracted graphs, then in reverse order divide and expand the contracted graphs to get a recursive division. The first approach relies on specific properties of planar graphs and maintaining spanning trees dynamically. It seems infeasible to use similar ideas on many other graph classes. The second approach relies on edge contraction and the minor-closed property of planar graphs. However, Euclidean graph classes are not minor-closed, yet edge contraction in the Euclidean space seems meaningless. 

We compute a recursive division of the Euclidean graph classes by (1) performing edge contractions on the Euclidean graph combinatorially rather than geometrically and (2) proving that the combinatorially contracted graphs have sublinear separators that can be computed in linear time. To prove that the combinatorially contracted graphs have sublinear separators, for each contracted graph we construct an auxiliary Euclidean graph which we call the \emph{representative graph} of the contracted graph. A representative graph of the contracted graph is a sampled subgraph of the original Euclidean graph. Rather than using the contracted graphs, we use (properties of) their representative graphs to prove that the contracted graphs have sublinear separators. 

Moreover, to compute an $r'$-division of the contracted graph (for different contracted graphs the parameter $r'$ is different), it is normally required that any $k'$-vertex subgraph of the contracted graph has a sublinear separator, for any $k'$. This is because when dividing the contracted graphs into smaller and smaller pieces, it is required that the smaller pieces still have sublinear separators. However, the contracted graphs of the Euclidean graph classes do not have this property, even if we use their representative graphs to find the separators. Let the contracted graph be $G_i$ and let the accumulated contraction size be the maximum number of vertices in $G$ that are contracted into a vertex in $G_i$. We resolve this critical issue by requiring the division size $r'$ to be some constant power of the accumulated contraction size. This requirement forces us to use different parameters for contractions and divisions. This is in contrast to the case in HKRS~\cite{SSSP-planar_HenzingerKRS1997}. Despite this modification, we prove that the effect is sufficient by showing that the induced recursive division is still suitable for applying HKRS's edge-relaxation algorithm in linear time.

\subsection{Related work}\label{ssec:related-work}
For graphs with non-negative edge weights, there exists algorithms running faster than the $O(m+n\log n)$ time by Dijsktra's algorithm with Fibonacci heaps. Pettie and Ramachandran~\cite{sssp-bounded-weight-ratio_PettieR05} gave an $O(m\alpha(m,n)+\min\{n\log n,n\log\log r\})$ time SSSP algorithm, where $\alpha(m,n)$ is the Inverse Ackermann function and $r$ is the ratio between maximum and minimum edge weights. Duan, Mao, Shu and Yin~\cite{sssp-general-bundled-Dijkstra_Duan-focs23} gave a randomized $O(m\sqrt{\log n \log \log n})$ time SSSP algorithm for graphs with non-negative edge weights. Very recently, Duan \etal~\cite{breaking-sorting-barrier-sssp-directed_DuanEtal-stoc25} give a deterministic $O(m\log ^{2/3} n)$ time SSSP algorithm for \emph{directed} graphs with non-negative edge weights. When the non-negative edge weights are restricted to integers, heaps with an $o(\log n)$ time \emph{delete-min} operation are known, such as the AF-heap~\cite{trans-dichotomous-sssp-FredmanWillard-jcss94} and the monotone priority queue~\cite{ram-priority-queue_Thorup}. Using these heaps, improvements on Dijkstra's algorithm can be obtained~\cite{trans-dichotomous-sssp-FredmanWillard-jcss94,sssp-integer_Raman-sigact-97,float-integer-sssp_Thorup,ram-priority-queue_Thorup}. On the other hand, when the edge weights are integers and can be negative, almost linear time SSSP algorithms have been discovered~\cite{sssp-negative-weight_focs2023,sssp-negative-weight-near-linear_focs2022, mincost-flow-near-linear-time_Peng-focs2022}, while all these algorithms have at least logarithmic dependence on the magnitude of the most negative edge weight. Finally, for graphs with arbitrary real edge weights, the classical $O(mn)$ time Bellman-Ford algorithm remains the fastest until recent randomized algorithms~\cite{sssp-negative-real_Fineman-stoc2024,sssp-negative-real_HuangJQ-soda2025} broke the $O(mn)$ time barrier.

For planar graphs, single-source shortest path (SSSP) computation is closely related to network flow. If the planar graph is undirected, Frederickson gave an $O(n\log n)$ time algorithm for maximum $st$-flow. Italiano, Nussbaum, Sankowski and Wulff-Nilsen~\cite{min-cut-undirected-planar_ItalianoNSW11} broke Frederickson's $O(n\log n)$ bound, giving an $O(n\log\log n)$ time algorithm for maximum $st$-flow. Their algorithm uses the linear time algorithm~\cite{SSSP-planar_HenzingerKRS1997} for single-source shortest-path computation. If the planar graph is directed, Miller and Naor~\cite{flow-planar-graph_MillerN95} showed how to compute maximum $st$-flow by solving a sequence of single-source shortest-paths with negative lengths. Following this approach, Henzinger, Klein, Rao and Subramanian~\cite{SSSP-planar_HenzingerKRS1997} devised an $O(n^{4/3}\log nL)$ time algorithm for single-source shortest-path with negative lengths, and solves maximum $st$-flow in $O(n^{4/3}\log n\log C)$ time, where $C$ is the sum of edge capacities. Later Fakcharoenphol and Rao~\cite{negative-sssp-planar_FakcharoenpholR01} presented an $O(n \log^3 n)$ time algorithm for single-source shortest-path with negative lengths, which implies an $O(n\log ^3 n\log C)$ time bound for maximum $st$-flow. Finally Borradaile and Klein~\cite{max-flow-directed-planar_BorradaileK2009} gave an $O(n\log n)$ time algorithm for maximum $st$-flow in planar directed graphs. Their algorithm involves finding single-source shortest-path distances in the dual planar graph, interpreting flow capacities as distances.

For planar graphs, Miller and Naor~\cite{flow-planar-graph_MillerN95} also showed how to solve feasible flow and bipartite perfect matching by computing single-source shortest-path with negative lengths. To date, the fastest algorithm for planar single-source shortest-path with negative weights is the $O(n\log^2n/\log\log n)$ time algorithm in~\cite{shortest-path-planar-negative-weight_MozesWulff2010}. By using the algorithm, one can solve feasible flow and bipartite perfect matching in planar graphs in $O(n\log^2n/\log\log n)$ time.

Sublinear separators have important algorithmic implications, besides shortest path. Graph classes that are subgraph-closed and admit sublinear separators have \emph{polynomial expansion}. Many NP-hard optimization problems admit a polynomial-time approximation scheme (PTAS) in graphs with polynomial expansion~\cite{ptas-polynomial-expansion_HarPeled&Quanrud17,ptas-polynomial-expansion-thin-graphs_Dvorak-soda18}. Sublinear separators have been found in graphs of bounded genus~\cite{separator-bounded-genus-graphs_GilbertHT84}, $K_l$-minor-free graphs~\cite{separator-minor-free_Reed&Wood09}, geometric graphs with sublinearly many edge crossings~\cite{geometric-graph-sublinear-crossings_EppsteinGS-soda2009}, $k$-ply neighborhood systems~\cite{separator-k-ply-neighborhood_EppsteinMT-SCG93}, $\tau$-lanky graphs~\cite{lanky-graph-sublinear-separator_LeThan-soda2022} and several geometric intersection graphs~\cite{clique-separator-geom-intersection-graph_BergEtal2023}.

\subsection{Paper Organization}
In Section~\ref{sec:prelim} we give preliminaries that are required for later sections. In Section~\ref{sec:main-theorem} we prove the main theorem. In Sections~\ref{sec:sssp-lanky} we consider SSSP in $\tau$-lanky graphs. In Section~\ref{ssec:sssp-k-ply-neighborhood} we consider SSSP in the intersection graph of $k$-ply neighborhood system and in Section~\ref{ssec:sssp-kappa-thick-cubical} we consider SSSP in the intersection graph of $\kappa$-thick cubical neighborhood system. 

\section{Preliminaries}\label{sec:prelim}

\begin{definition}[$r$-division]
    \label{definition:division}
    Given any integer $r>0$, an $r$-division of an $n$-vertex graph divides the graph into $O(n/r)$ regions where each region has $O(r)$ vertices (edges) and $O(\sqrt{r})$ boundary vertices. A region consists of its vertices and edges. A vertex is called a boundary vertex if it is contained in multiple regions.
\end{definition}

\begin{figure}[htb]
    \centering
    \includegraphics[width=.6\textwidth]{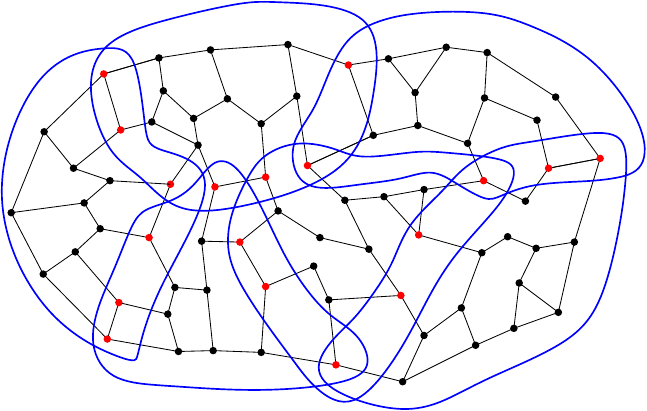}
    \caption{Illustration of $r$-division. An $r$-division of the graph divides the graph into regions (enclosed in blue lines) with interior vertices (black) and boundary vertices (red). Modified figure from~\cite{shortest-path-planar_Frederickson1987}.}
    \label{fig:r-division}
\end{figure}

\begin{definition}[$(r,s)$-division]
    \label{definition:rs-division}
    An $(r,s)$-division divides an $n$-vertex graph into $O(n/r)$ regions, such that each region has $r^{O(1)})$ vertices and $O(s)$ boundary vertices.
\end{definition}

A recursive division of a graph divides the graph into regions, then divides the regions into subregions, and so on recursively. Formally, an $(\bar{r},f)$-recursive division is defined as follows.
 
\begin{definition}[recursive division~\cite{SSSP-planar_HenzingerKRS1997}]
    \label{definition:recursive_division} Let $f$ be a non-decreasing positive integer function and $\bar{r}=(r_0,r_1,\ldots,r_k)$ be a positive integer sequence. An $(\bar{r},f)$-recursive division of a graph consists of
    \begin{enumerate}[(i)]
        \item an $(r_k,f(r_k))$-division of the graph, and
        \item an $((r_0,r_1,\ldots,r_{k-1}),f)$-recursive division for each region in the $(r_k,f(r_k))$-division.
    \end{enumerate}
\end{definition}


A recursive division can be computed in two phases, the \emph{contraction phase} and the \emph{division phase}. 
In the contraction phase, a sequence of non-decreasing integers $(z_0,z_1,\ldots,z_{I})$ are chosen and a sequence of consecutive contractions are performed. Let $G_0=G$ and let $\{G_i|1\leq i\leq I+1\}$ be the contracted graphs. Let $G_i$=Contract$(G_{i-1},z_{i-1})$ where the contraction algorithm Contract$(G_{i-1},z_{i-1})$ contracts $G_{i-1}$ into $G_i$. The Contract$(G_{i-1},z_{i-1})$ algorithm calls the clustering algorithm in~\cite{online-update-MST_Frederickson85} to cluster the vertices of $G_{i-1}$ into $O(n_{i-1}/z_{i-1})$ connected sets ($n_{i-1}$ is the number of vertices in $G_{i-1}$) with each set containing $O(z_{i-1})$ vertices, then contracts each connected set into a vertex of $G_i$. The contraction phase ends when the last contracted graph $G_{I+1}$ has at most $n/\log n$ vertices. 

\begin{definition}\label{def:contracted-graph}
    Let $G=G_0$ be a graph. We call $\{G_i|1\leq i\leq I+1\}$ the (combinatorially) contracted graphs of $G$, where $G_i=Contract(G_{i-1},z_{i-1})$ for some integer $z_{i-1}$.
\end{definition}

In the division phase, the contracted graphs $G_1,\ldots,G_{I+1}$ are considered in reverse order and a division is obtained for each of the contracted graphs. Let $R$ be an $n'$-vertex region with a set $S$ of boundary vertices. The Divide($R,S,r$) algorithm calls the separator algorithm recursively and computes an $(r,s)$-division of $R$ (which divides $R$ into subregions). A vertex in a subregion of $R$ is a boundary vertex if it is contained in multiple subregions or it is in $S$. The division $D_{I+1}$ of $G_{I+1}$ consists of one region, i.e., $G_{I+1}$. Starting from $i=I$, the division $D_i$ of $G_i$ is obtained from the division $D_{i+1}$ of $G_{i+1}$ by (i) calling Divide($R_{i+1},S_{i+1},z_i$) for each region $R_{i+1}$ of $D_{i+1}$ and (ii) expanding each vertex of the subregions in $R_{i+1}$ to the connected set in $G_i$ that is contracted into the vertex during the contraction phase. Thus a tree representing the recursive division, called the \emph{recursive division tree}, is obtained.
\vspace{\baselineskip}

\noindent\textbf{HKRS's edge-relaxation algorithm.} Inspired by Frederickson's work, Henzinger, Klein, Rao and Subramanian~\cite{SSSP-planar_HenzingerKRS1997} developed an optimal $O(n)$ time SSSP algorithm for planar graphs with non-negative edge weights. The algorithm consists of two parts. The first part computes a recursive division of the planar graph in linear time. The second part runs an edge-relaxation algorithm on the recursive division to compute the SSSP. Each node in the recursive division tree represents a region/subregion in the division. Each node (region/subregion) is associated with a heap, whose elements are (the current) distance values to some vertices in the region. The elements in the heap of a region are the min-elements of its subregions' heaps. This forms a hierarchy of heaps. Inside a region, edges are relaxed in analogy to Dijkstra's algorithm, while distance values in the associated heaps are updated. In contrast to Dijkstra's algorithm, the algorithm does not perform edge relaxations for each edge of the region, but only performs a certain number of edge relaxations and then jumps to another region. The numbers of edge relaxations performed inside the regions are carefully coordinated so that (i) when the algorithm terminates, the distance labels of the vertices are the exact distances from the source vertex, (ii) heap operations on large heaps (which are expensive) can be charged to heap operations on small heaps so that the overall heap operations take only $O(n)$ time. In this way the edge-relaxation algorithm correctly computes the distance labels and takes only $O(n)$ time.

A main contribution of HKRS's work is the proof of the following fact. Note that the fact holds for any graph as long as a recursive division satisfying the conditions is given.

\begin{fact}[Section~3 in~\cite{SSSP-planar_HenzingerKRS1997}]
    \label{fact:HKRS_section3}
    Suppose we are given a graph and an $((r_0,r_1,\ldots,r_k),f)$-recursive division of the graph satisfying
\begin{align}\label{eqn:tricky-condition}
    \frac{r_i}{f(r_i)}\geq 8^if(r_{i-1})\log r_{i+1} \left ( \sum^{i+1}_{j=1}\log r_j\right )
\end{align}

    for all $r_i$'s exceeding a constant. Then applying the edge-relaxation algorithm on the recursive division takes time linear to the size of the graph.
\end{fact}

\section{The main theorem}\label{sec:main-theorem}

In this section, let $G$ be a Euclidean graph in $\R^d$ that satisfies the following \textbf{criteria}:
\begin{enumerate}[(I)]
    \item $G$ admits an $O(c_1 n^{1-1/d})$ balanced separator that can be computed in $O(c_2 n)$ time. \label{main-crit1}
    \item The combinatorially contracted graphs $G_i$ of $G$ are sparse, i.e., any $\bar{k}$-vertex subgraph of $G_i$ has $O(c_3 \bar{k})$ edges. \label{main-crit2}
    
    \item A closed surface $\xi$ in $\R^d$ (such as a Jordan curve in $\R^2$, a sphere in $\R^d$) is used to find the separator in (I). The closed surface cuts $O(c_4 n^{1-1/d})$ geometrical objects (such as segments, balls) associated with the vertices in $G$, and these vertices form the separator. Moreover, a random closed surface cuts $O(c_4 n^{1-1/d})$ geometrical objects (such as segments, balls) associated with the vertices in $G$ in expectation. \label{main-crit3}
    \item The graph class of $G$ is subgraph-closed. \label{main-crit4}
\end{enumerate} 
Here $c_1$, $c_2$, $c_3$ and $c_4$ are constants. In criteria (I) and (III) we set the exponent of $n$ to be $1-1/d$ for ease of exposition. Adapting to any constant $< 1$ is immediate.

As mentioned in Section~\ref{ssec:tech-review}, we will use the edge-relaxation algorithm of HKRS once a recursive division satisfying Inequalities~\eqref{eqn:tricky-condition} is given. HKRS's edge-relaxation algorithm assumes that the input graph has maximum degree at most 3. Hence, we first transform $G$ into a graph with maximum degree 3. We perform the following transformations on the Euclidean graph $G$. For each vertex $v$ of degree $D$ greater than $3$, let $u_1,\ldots,u_D$ be a cyclic ordering of the vertices adjacent to~$v$. Replace $v$ by a cycle $v_1,\ldots,v_D,v_1$ where each of the vertices $v_i$ ($1\leq i\leq D$) is arbitrarily close to $v$. Replace edge $(u_i,v)$ with new edge $(u_i,v_i)$. Assume that $G$ has $O(c_0 n)$ edges ($c_0$ is constant), the transformed graph has $O(c_0 n)$ vertices whose maximum degree is 3. 

The main task is to compute, in linear time, a recursive division of $G$ that satisfies Inequalities~\eqref{eqn:tricky-condition}. To compute this recursive division, we compute a sequence of combinatorially contracted graphs of $G$ and construct the auxiliary representative graphs of the contracted graphs. In Section~\ref{ssec:contracted-graph-represent}, we use the sparsity of the contracted graphs $G_i$ (criterion (II)) to construct their representative graphs. In Section~\ref{ssec:separator-contracted-graphs} we prove that the (combinatorially) contracted graphs have sublinear separators and such sublinear separators can be found efficiently using their representative graphs. Finally, in Section~\ref{ssec:recursive-division-main}, we prove that a recursive division of $G$ satisfying Inequalities~\eqref{eqn:tricky-condition} can be computed in linear time and conclude the linear time SSSP algorithm.

\subsection{The contracted graphs and their representative graphs}\label{ssec:contracted-graph-represent}
Recall that the recursive division can be computed using a two-phase algorithm. In the contraction phase of the recursive division algorithm, the contracted graphs $G_1,G_2,\ldots,G_{I+1}$ are constructed consecutively such that $G_i=Contract(G_{i-1},z_{i-1})$, $1\leq i\leq I+1$, where Contract$(G_{i-1},z_{i-1})$ is the contraction algorithm. For each contracted graph $G_i$, $1\leq i\leq I+1$, we construct its representative graph $RG_i$.

Each edge of $G_i$ has a \emph{representative edge} in $RG_i$, which is an edge in $G$. The representative graph $RG_i$ consists of the representative edges of the edges in $G_i$ and their endpoints. The representative edge of an edge $e$ in $G_i$ is denoted as $rep_{G_i}(e)$. The representative edge of an edge in $G_0=G$ is the edge itself, and the representative graph of $G_0$ is $G_0$.

The representative graph $RG_i$ is constructed from the contracted graph $G_{i-1}$ and its representative graph $RG_{i-1}$, as follows. Consider a vertex $u$ in $G_i$ and let $cset_{G_{i-1}}(u)$ denote the connected set of vertices in $G_{i-1}$ from which $u$ is contracted. An edge $e=(u,v)$ in $G_i$ corresponds to the set of edges in $G_{i-1}$ between $cset_{G_{i-1}}(u)$ and $cset_{G_{i-1}}(v)$. 

Let $REP_{G_i}(e)=\{rep_{G_{i-1}}((u',v'))|u'\in cset_{G_{i-1}}(u),v'\in cset_{G_{i-1}}(v), (u',v')\in E(G_{i-1})\}$ denote the set of the representative edges of the edges in $G_{i-1}$ between $cset_{G_{i-1}}(u)$ and $cset_{G_{i-1}}(v)$. See Figure~\ref{fig:representative-graph}(a) for an illustration. For each edge $e$=$(u,v)$ in $E(G_{i})$, choose an arbitrary edge in $REP_{G_i}(e)$ as its representative edge. The chosen representative edges and their endpoints form the representative graph $RG_i$ of $G_i$.

\begin{figure}[htb]
    \centering
    \includegraphics[width=.85\textwidth]{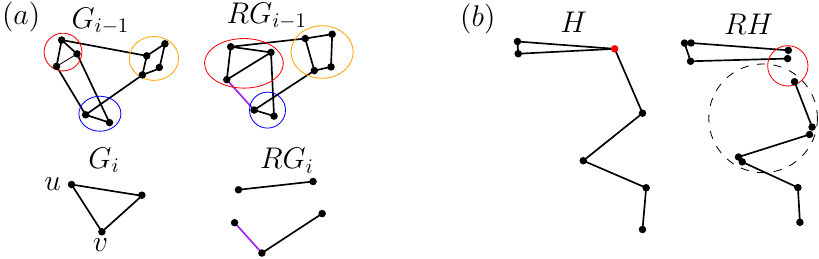}
    \caption{(a) $cset_{G_{i-1}}(u)$ is enclosed in red, $cset_{G_{i-1}}(v)$ is enclosed in blue. The representative edge of $(u,v)$ in $G_i$ is the purple edge. (b) Illustrating the proof of Lemma~\ref{lem:separator-contracted-graph}. The closed surface $\xi^*$ is the dashed circle. The red vertex is added to $S$ since its $\Gamma(\cdot)$ has one point inside $\mathbf{b}(o,r^*)$ and two points outside $\xi^*$. }
    \label{fig:representative-graph}

\end{figure}

According to criterion (II), $RG_{i-1}$ has $O(c_3 n_{i-1})$ edges and vertices. It is then observed that $RG_i$ is constructed from $G_{i-1}$ and $RG_{i-1}$ in $O(c_3 n_{i-1})$ time where $n_{i-1}$ is the number of vertices in $G_{i-1}$. Note that $RG_i$ may not be connected.
\begin{observation}\label{obs:RGi-main}
    The representative graph $RG_i$ of $G_i$ can be constructed in $O(c_3 n_{i-1})$ time. 
\end{observation}

\subsection{Separators of the contracted graphs}\label{ssec:separator-contracted-graphs}
The parameters $z_i$ used in the contraction phase, $0\leq i\leq I$, are set as in~\cite{SSSP-planar_HenzingerKRS1997}:
\begin{align} 
    z_0&=2,\label{eqn:z_0-lanky}\\
    z_{i}&=7^{z_{i-1}^{1/5}},1\leq i\leq I.\label{eqn:z_i-lanky}
\end{align} 
Let $\alpha_{i-1}=\prod\limits_{j< i}z_j$ be the accumulated contraction sizes for the contracted graph $G_i$, and let $\beta_{i-1}=z_{i-1}^{3d-2}$. In~\cite{SSSP-planar_HenzingerKRS1997}, parameters $z_i$ ($0\leq i\leq I$) are used for division in the division phase. In contrast, we use $\beta_{i}$ in place of $z_i$ as division parameters in the division phase. 

A subgraph of $G_i$ with $O(\alpha_{i-1})$ vertices may not admit a sublinear separator. However, the next lemma shows that any subgraph of $G_i$ with $\Omega(\beta_{i-1})$ vertices admit a sublinear separator.

\begin{lemma}\label{lem:separator-contracted-graph}
    Let $H$ be any subgraph of $G_i$ with $k=\Omega(\beta_{i-1})$ vertices. A balanced separator of $H$ of size $O((c_1c_3^{1-1/d}+c_4)k^{(1-\frac{1}{3d-2})})$\footnote{We did not try to optimize the size of the separator.} can be computed in $O(c_2c_3 k)$ time given $H$ and $RH$. 
\end{lemma}
\begin{proof}
    Let $RH$ denote the representative graph of $H$, which is a subgraph of $RG_i$. According to criterion~(II), $RH$ has $O(c_3 k)$ vertices. Run the separator algorithm on $RH$ (criterion (I)), noting that $RH$ is a subgraph of $G$, thus it also admits sublinear separator (criteria (I) and (IV)). Let $\xi^*$ be the closed surface (criterion (III)) used to find the separator of $RH$. Let $S$ be the desired balanced separator of $H$, initially being empty. Vertices are added to $S$ in two steps. 
    \begin{enumerate}[(1)]
        \item For any vertex $v$ in $G_i$, let $conset_G(v)$ be the (accumulative) connected set of vertices in $G$ that are contracted into $v$ in the contraction phase. For any point $p$ in $RG_i$, let $rev(p)$ be the vertex $v$ in $G_i$ such that $p\in conset_G(v)$. Note that there is exactly one such vertex $v$. In the first step, for each geometrical object (criterion (III)) that is cut by $\xi^*$, add its associated $rev(\cdot)$ to $S$. For example, if the geometrical object is a ball and $p$ is the ball's center, add $rev(p)$ to $S$.

        \item For any vertex $v$ in $G_i$, let $\Gamma(v)$ be (the set of) endpoints of edges in $RG_i$ that belong to $conset_G(v)$ (we have $\Gamma(v)\subset conset_G(v)$). In the second step, partition the vertices of $H$ into three groups: (i) those such that all points in its $\Gamma(\cdot)$ lie inside $\xi^*$, (ii) those such that all points in its $\Gamma(\cdot)$ lie outside $\xi^*$, and (iii) those such that some points in its $\Gamma(\cdot)$ lie inside $\xi^*$ and some points in its $\Gamma(\cdot)$ lie outside $\xi^*$. Add vertices in group (iii) to $S$. See Figure~\ref{fig:representative-graph}(b) for an illustration.
    \end{enumerate}
    

     We prove that $S$ is a balanced separator of $H$ with $O((c_1c_3^{1-1/d}+c_4)k^{(1-\frac{1}{3d-2})})$ vertices. We first prove the size of $S$. Since $RH$ is a subgraph of $G$ with $O(c_3 k)$ edges and vertices, the closed surface $\xi^*$ cuts $O(c_1 (c_3 k)^{1-1/d})=O(c_1c_3^{1-1/d}k^{1-1/d})$ geometrical objects associated with the vertices of $RG_i$. Since $\Gamma(v)$ is a subset of $conset_G(v)$ and the vertices in $conset_G(v)$ are connected in $G$, the vertex $v$ belongs to group (iii) only if one or more geometrical objects in (associated with) the subgraph of $G$ induced by vertices in $conset_G(v)$ is cut by $\xi^*$. Let $G_H$ be the subgraph of $G$ from which $H$ is contracted. $G_H$ has $O(\alpha_{i-1} k)$ vertices where $\alpha_{i-1}=\prod_{j<i}z_j$. Since $\xi^*$ is a random closed surface in $G_H$, according to criterion (III), $\xi^*$ cuts $O(c_4(\alpha_{i-1} k)^{1-1/d})$ geometrical objects associated with the vertices of $G_H$ in expectation. Since $z_{i}=7^{z_{i-1}^{1/5}}$, $\alpha_{i-1}=O(\beta_{i-1}^{\frac{1}{3d-2}})=O(k^{\frac{1}{3d-2}})$. Thus $\xi^*$ cuts order of
     \begin{align}
         c_4 (\alpha_{i-1} k)^{1-1/d}&\leq c_4 k^{(1+\frac{1}{3d-2})(1-1/d)}\\
         &=c_4 k^{1-\frac{2d-1}{d(3d-2)}}\leq c_4 k^{1-\frac{d}{d(3d-2)}}=c_4 k^{1-\frac{1}{3d-2}}
     \end{align}
     geometrical objects associated with $G_H$. The number of vertices in group (iii) is thus $O(c_4 k^{(1-\frac{1}{3d-2})})$. Therefore, the size of $S$ is $O((c_1c_3^{1-1/d}+c_4)k^{(1-\frac{1}{3d-2})})$. 

     Next, we prove that $S$ is a balanced separator. Let $A$ be the set of vertices in $G_i$ which belong to group (i) and are not added to $S$. Let $B$ be the set of vertices in $G_i$ which belong to group (ii) and are not added to $S$. We claim that the removal of $S$ separates $A$ and $B$. We prove the claim by contradiction.
     
     Each edge in $H$ corresponds to a unique edge in $RH$. Each vertex $v$ in H corresponds to a set of vertices $\Gamma(v)$ in $RH$. Let $a$ be a vertex in $A$ and let $b$ be a vertex in $B$. By the definition of $A$ and $B$, $a$ is in group (i) so all vertices in $\Gamma(a)$ are inside $\xi^*$ and $b$ is in group (ii) so all vertices in $\Gamma(b)$ are outside $\xi^*$. If there is an edge in $H$ between $a$ and $b$, then there is a unique edge $(a_r,b_r)$ in $RH$ such that $a_r\in \Gamma(a)$ and $b_r\in \Gamma(b)$. Since $a_r$ is inside $\xi^*$ and $b_r$ is outside $\xi^*$, either the geometrical object associated with $a_r$ or the geometrical object associated with $b_r$ would be cut by $\xi^*$, which implies that either $a$ or $b$ were added to $S$. A contradiction.
     
     Since $\xi^*$ separates vertices of $RH$ in a balanced way and $|S|$ is sublinear in $k$, $|A|$ and $|B|$ are proportional to each other. Therefore, $S$ is a balanced separator. 
     
     According to criterion (I), running the separator algorithm on $RH$ takes $O(c_2c_3 k)$ time. This finishes the proof of the lemma.  
\end{proof}

\subsection{Recursive division}
\label{ssec:recursive-division-main}
The recursive division of $G$ is computed in two phases. In the contraction phase, the contraction parameters $z_i$ are set as Equations~\eqref{eqn:z_0-lanky} and~\eqref{eqn:z_i-lanky}, and the contracted graphs $G_1,\ldots,G_{I+1}$ are computed consecutively by the Contract() algorithm. In the division phase, we choose the parameters for dividing $G_i$, $1\leq i\leq I+1$, to be $\beta_{i-1}=z_{i-1}^{3d-2}$. Since the minimum subgraph of $G_i$ that will be divided has $\Omega(\beta_{i-1})$ vertices, we can apply Lemma~\ref{lem:separator-contracted-graph} recursively to obtain the division $D_i$ of $G_i$.

To analyze the computing time of the recursive division, we analyze the division $D_i$ of $G_i$, and the division of $G$ it induces. Recall that the division $D_i$ is derived from the division $D_{i+1}$ of $G_{i+1}$. Let $R_{i+1}$ be a region in $D_{i+1}$. The Divide$(R_{i+1},S_{i+1},\beta_{i})$ algorithm divides $R_{i+1}$ into subregions each of which has at most $\beta_i=z_i^{3d-2}$ vertices and has $O((c_1c_3^{1-1/d}+c_4)z_i^{3d-3})$ boundary vertices. Each expanded subregion, which is a region in $D_i$, has at most $\beta_i\cdot 3z_i=3z_i^{3d-1}$ vertices and $O((c_1c_3^{1-1/d}+c_4)z_i^{3d-2})$ boundary vertices. Let $k_i$ denote the number of regions in the division $D_i$, $0\leq i\leq I+1$, and let $n_i$ be the number of vertices in $G_i$. We can prove, using similar arguments as in the proof of Lemma~4.1 in~\cite{SSSP-planar_HenzingerKRS1997} that
\begin{lemma}\label{lem:region-number-Di_main}
    The number of regions in the division $D_i$ is $O(n_i/z_i^{3d-1})$.
\end{lemma}
From the proof of Lemma~\ref{lem:region-number-Di_main}, we can also prove that the sum 
of the vertices of all regions in $D_{i}$ is $O(n_{i})$. Now we can analyze the computing time of the recursive division of $G$.
\begin{lemma}\label{lem:division-runtime-main}
    Let $G$ be an Euclidean graph in $\R^d$ that satisfies the criteria (I)-(IV) and has $O(c_0 n)$ edges. A recursive division of $G$ can be computed in $O(c_0c_2c_3d n)$ time.
\end{lemma}
\begin{proof}
We calculate the costs in steps.
    \begin{enumerate}
        \item The number of edges in $G_i$, $i\geq 1$, is $O(c_3 n_i)$ by criterion (II). Thus, computing the contracted graphs $G_1,\ldots,G_{I+1}$ takes time $O(c_0 n+c_3n_1+\ldots+ c_3n_I)$, which is $O(c_0c_3 n)$.
        \item By Observation~\ref{obs:RGi-main}, computing the representative graphs $RG_1,\ldots,RG_{I+1}$ takes $O(c_0 n+ c_3 n_1+\ldots+ c_3 n_I)=O(c_0c_3 n)$ time.
        \item For $i\leq I$, divide a region $R_{i+1}$ in the division $D_{i+1}$ of $G_{i+1}$ takes $O(c_2c_3|R_{i+1}|\log |R_{i+1}|)$ time, by recursively applying Lemma~\ref{lem:separator-contracted-graph}. 
        \item Each region $R_{i+1}$ in $D_{i+1}$ has $O(z_{i+1}^{3d-1})$ vertices. Summing over all regions in $D_{i+1}$, dividing $D_{i+1}$ takes $O(c_2c_3 d\sum_{R_{i+1}}|R_{i+1}|\log z_{i+1})$ time, which is $O(c_2c_3 d n_{i+1}\log z_{i+1})$ by the fact that $\sum_{R_{i+1}}|R_{i+1}|$ is $O(n_{i+1})$.
        \item The time to obtain the division $D_i$ of $G_i$ after dividing $D_{i+1}$ is $O(c_3 n_i)$. Computing all the divisions $D_i$, $1\leq i\leq I+1$, thus takes $O(c_2c_3d\sum_i n_{i}\log z_{i})$ time. Since $n_{i}\leq n_{i-1}/z_{i-1}\leq c_0n/z_{i-1}$ and $\log z_{i}$ is $O(z_{i-1}^{1/5})$, $O(c_2c_3d\sum_i n_{i}\log z_{i})$ is $O(c_2c_3d c_0n)=O(c_0c_2c_3d n)$. 
    \end{enumerate}

    Therefore, the total running time is $O(c_0c_2c_3d n)$, which finishes the proof of the lemma.
\end{proof}

Finally, we prove that the recursive division satisfies Inequalities~\eqref{eqn:tricky-condition}, which will imply that applying HKRS's edge-relaxation algorithm on it will take $O(c_0 n)$ time. 
\begin{lemma}\label{lem:valid-division-main}
    The recursive division of $G$ satisfies Inequalities~\eqref{eqn:tricky-condition}.
\end{lemma}
\begin{proof}
    Since $n_{j}\leq n_{j-1}/z_{j-1}$, we have
        $n_i\leq c_0 n/\prod_{j<i}z_j$.    
    Note also that each vertex in $G_i$ expands to at most $\prod_{j<i}3z_j$ vertices of $G$. 

    Consider the division $D_i$ of $G_i$, and the division of $G$ it induces. The division $D_i$ consists of $O(n_i/z_i^{3d-1})$ regions (Lemma~\ref{lem:region-number-Di_main}), each of which has $O(z_i^{3d-1})$ vertices and $O((c_1c_3^{1-1/d}+c_4) z_i^{3d-2})$ boundary vertices. Thus it induces a division of $G$ consisting of $O(n_i/z_i^{3d-1})$ regions, each of which has $O(z_i^{3d-1}\prod_{j<i}3z_j)$ vertices and $O((c_1c_3^{1-1/d}+c_4)z_i^{3d-2}\prod_{j<i}3z_j)$ boundary vertices.

    Let $r_i=z_i^{3d-1}\prod_{j<i}z_j$, and let 
    \begin{align*}
        f(r_i)=c'(c_1c_3^{1-1/d}+c_4) z_i^{3d-2}\prod_{j<i}z_j.
    \end{align*}
    The induced division of $G$ has $O(c_0n/r_i)$ regions each having $O(r_i3^i)$ vertices and $O(f(r_i))$ boundary vertices. Since $3^i\leq \prod_{j\leq i}z_j$, each region has $O(r_i^2)$ vertices.

    We have
        $$\frac{r_i}{f(r_i)}=\frac{z_i^{3d-1}}{c'(c_1c_3^{1-1/d}+c_4) z_i^{3d-2}}
        =\frac{z_i}{c'(c_1c_3^{1-1/d}+c_4)}.$$

    By the setting of $z_i$, we have $z_{i-1}=\Theta(\log ^5 z_i)$ and $\prod_{j<i}z_j=O(\log ^6 z_i)$. Therefore
        $$f(r_{i-1})=c'(c_1c_3^{1-1/d}+c_4) z_{i-1}^{3d-2}\prod_{j<i-1}z_j
                  =O((c_1c_3^{1-1/d}+c_4) \log^{5(3d-2)}z_i\log^6 \log z_i).$$
    We have
        $\log r_{i+1}=O(\log z_{i+1}^{3d-1})=O(d\cdot z_i^{1/5})$
    and $\sum_{j=1}^{i+1}\log r_j=O(d\cdot z_i^{1/5})$. Since $c_1,c_3,c_4,d$ are all constants, for a sufficiently large constant $\hat{c}$, Inequalities~\eqref{eqn:tricky-condition} hold for all $r_i$ exceeding~$\hat{c}$.
\end{proof}

We have obtained the main theorem of this section.
\begin{theorem}\label{thm:recursive-division-main}
Let $G$ be an Euclidean graph in $\R^d$ ($d$ is fixed) that satisfies criteria (I)-(IV) and has $O(c_0 n)$ edges. Then a $(\bar{r},f)$-recursive division of $G$ satisfying Inequalities~\eqref{eqn:tricky-condition} can be computed in $O(c_0c_2c_3 dn)$ time.
\end{theorem}
As a corollary, 
\begin{corollary}\label{col:sssp-main}
Let $G$ be an Euclidean graph in $\R^d$ ($d$ is fixed) that satisfies criteria (I)-(IV) and has $O(c_0 n)$ edges. Single-source shortest path in $G$ can be computed in $O(c_0c_2c_3 dn)$ time.
\end{corollary}
Note that since edge-relaxation is comparison-based and does not make any assumptions about edge weights, the edge weights of the Euclidean graph can be any nonnegative value, not just the Euclidean distance between the endpoints.

\section{SSSP for lanky-graphs}\label{sec:sssp-lanky}
Lanky graphs were introduced by Le and Than~\cite{lanky-graph-sublinear-separator_LeThan-soda2022} when studying sublinear separators of greedy spanners in Euclidean and doubling metrics. 

\begin{definition}
\label{def:lanky-graph}
    A graph $G$ embedded in an Euclidean or a doubling metric is $\tau$-lanky, if for any $r>0$ and any ball of radius $r$, there are at most $\tau$ edges of length at least $r$ that are cut by the ball. An edge is cut by a ball if one endpoint of the edge is inside the ball and the other endpoint of the edge is outside the ball.
\end{definition}

A number of graphs are $O(1)$-lanky graphs, including greedy spanners in Euclidean and doubling metrics, low-density graphs in Euclidean~\cite{map-matching-low-density_ChenDGNW-ALENEX2011} and doubling metrics, greedy spanners for point sets in $\R^d$ of low fractal dimension and greedy spanners for unit ball graphs in $\R^d$.

To unify the construction of sublinear separators for lanky graphs in both Euclidean and doubling metrics, Le and Than proposed the $(\eta,d)$-packable metric space. We say that a set of points $P$ is $r$-separated if the distance between any two points in $P$ is at least $r$.

\begin{definition}\label{def:packable-metric-space}
    A \emph{metric} $(X,\delta_X)$ is $(\eta,d)$-packable if for any $r\in (0,1]$ and any $r$-separated set $P\subseteq X$ contained in a unit ball, $|P|\leq \eta (\frac{1}{r})^d$. We call $d$ the packing dimension of the \emph{metric} and $\eta$ the packing constant of the \emph{metric}.
\end{definition}

The authors proved the following separator theorem for lanky graphs in an $(\eta,d)$-packable metric space. 
\begin{theorem}[\cite{lanky-graph-sublinear-separator_LeThan-soda2022}]\label{thm:separator-lanky-packable-space}
    Let $(X,\delta_X)$ be an $(\eta,d)$-packable metric space and let $G=(V,E,w)$ be a $\tau$-lanky graph in $(X,\delta_X)$. $G$ has a $(1-\frac{1}{\eta2^{d+1}})$-balanced separator $S$ such that $S$ has size $O(\tau\eta8^d n^{1-1/d})$. $S$ can be computed in $O((\eta^38^d+\tau)n)$ expected time. 
\end{theorem}

The separator algorithm in the proof of Theorem~\ref{thm:separator-lanky-packable-space} consists of two steps.
\begin{enumerate}
    \item\label{step:find-annulus} Compute a ball $\mathbf{b}(v,r)$ such that $\mathbf{b}(v,r)$ contains at least $\frac{n}{\eta2^{d+1}}$ vertices of $G$ and the concentric ball $\mathbf{b}(v,2r)$ contains at most $\frac{n}{2}$ vertices of $G$, using the randomized algorithm in~\cite{fast-construction-nets-low-metrics_HarPeled&Mendel2006}. 

    \item Choose $r^*=(1+\sigma)r$ where $\sigma\in (0,1]$ is chosen uniformly at random. Let $\mathbf{b}(v,r^*)$ be the random ball. Compute the set of edges $E^*$ of $G$ that are cut by $\mathbf{b}(v,r^*)$.  Return the endpoints of the edges in $E^*$ as the separator. 
\end{enumerate}

With constant probability, $E^*$ has size $O(\tau\eta8^d n^{1-1/d})$. This is proved by partitioning the edges in $E^*$ into groups of exponentially increasing lengths and summing over the groups. From the proof of Theorem~\ref{thm:separator-lanky-packable-space}, it can be inferred that a random ball $\mathbf{b}(x,y)$ cuts $O(\tau\eta8^d n^{1-1/d})$ edges of $G$ in expectation.

Until now, we have verified that a $\tau$-lanky graph in an $(\eta,d)$-packable metric space satisfies criteria (I) and (III) in Section~\ref{sec:main-theorem} with $c_1=\tau\eta8^d$, $c_2=\eta^38^d+\tau$ and $c_4=\tau\eta8^d$. The closed surface $\xi$ used to find the separator is a ball. The geometrical objects cut by $\xi$ are segments, which are edges in $G$. 

From the definition of lanky graphs, it is easy to verify they are subgraph-closed. Thus criterion (IV) is satisfied. It remains to verify criterion (II). 


\subsection{Sparsity of the contracted graphs}\label{ssec:sparsity-contracted-lanky}
We relate the sparsity of the combinatorially contracted graphs of a lanky graph to the graph-theoretic property \emph{thickness}. The thickness of a graph $G'$ is the minimum number of planar subgraphs that $G'$ can be decomposed into. Let $\theta(G')$ denote the thickness of a graph $G'$. Determining the thickness of a general graph is NP-hard~\cite{thickness-determination-NP-hard}. However, for graphs with maximum degree~$\Delta$, Halton~\cite{thickness-of-bounded-degree_Halton91} proved that if $G'$ is a graph with maximum degree $\Delta$ then $\theta(G')\leq \lceil \Delta/2 \rceil$. It follows from the definition of lanky graph that a $\tau$-lanky graph $G$ has maximum degree $\tau$. Therefore the thickness of $G$ is at most $\lceil\tau/2\rceil$.

Thus, the $\tau$-lanky graph $G$ can be decomposed into $\lceil\tau/2\rceil$ planar subgraphs, which we call $H_1,\ldots,\allowbreak H_{\lceil\tau/2\rceil}$. Let each of $H_1,\ldots,H_{\lceil\tau/2\rceil}$ have $V(G)$ as its vertex set. Use the clustering inside the Contract$(G,z_0)$ algorithm to cluster vertices in $V(G)$. For each $H_j$, $1\leq j\leq \lceil\tau/2\rceil$, shrink each cluster into a vertex while retaining edges between the clusters (shrunk vertices). Let Shrink$(H_j,z_0)$ be the shrunk graph of $H_j$. Then Contract$(G,z_0)$ is the union of Shrink$(H_j,z_0)$, $1\leq j\leq \lceil\tau/2\rceil$ (removing multi-edges).

Since Shrink$(H_j,z_0)$ is a planar graph, $G_1=Contract(G,z_0)$ has $O(\tau n_1)$ edges where $n_1$ is the number of vertices in $G_1$. Similarly $G_i$, $i\geq 2$, is the union of the iteratively shrunk graphs of $H_j$, $1\leq j\leq \lceil\tau/2\rceil$, thus has $O(\tau n_i)$ edges where $n_i$ is the number of vertices in $G_i$.

\begin{lemma}\label{col:sparsity-contracted-graph-lanky}
    The contracted graph $G_i$ of $G$, $1\leq i\leq I+1$, has $O(\tau n_i)$ edges.
\end{lemma}
Following similar ideas, we get:

\begin{corollary}\label{col:sparsity-contracted-subgraph-lanky}
    If $H$ is a subgraph of $G_i$ with $\bar{k}$ vertices then $H$ has $O(\tau \bar{k})$ edges.
\end{corollary}
This justifies that a $\tau$-lanky graph satisfies criterion (II) with $c_3=\tau$. 

Since $G$ has maximum degree $\tau$, $G$ has $O(\tau n)$ edges so $c_0=\tau$. Plugging the constants into Theorem~\ref{thm:recursive-division-main}, we have
\begin{theorem}\label{thm:recursive-division-lanky}
Let $G$ be a $\tau$-lanky graph in an $(\eta,d)$-packable metric space. A $(\bar{r},f)$-recursive division of $G$ satisfying Inequalities~\eqref{eqn:tricky-condition} can be computed in $O((\tau^2\eta^38^d + \tau^3) dn)$ randomized time.
\end{theorem}
As a corollary, we get
\begin{corollary}\label{col:sssp-lanky}
Let $G$ be a $\tau$-lanky graph in an $(\eta,d)$-packable metric space. Single-source shortest path in $G$ can be computed in $O((\tau^2\eta^38^d + \tau^3) dn)$ time.
\end{corollary}

\section{SSSP for the intersection graph of \texorpdfstring{$k$}{k}-ply neighborhood system}\label{ssec:sssp-k-ply-neighborhood}
In this section we consider single-source shortest path in the edge-weighted intersection graph of a $k$-ply neighborhood system in $\R^d$. Miller, Teng, Thurston and Vavasis~\cite{geometric-approach-graph-separator_focs91,separator-sphere-packing_MillerEtal-jacm1997, separator-finite-element-meshes_MillerEtal1998} developed a geometric characterization of graphs that have a small separator, based on the notion of a \emph{$k$-ply neighborhood system}. A $k$-ply neighborhood system is a collection of $n$ balls in $\R^d$ such that no point in the space is covered by more than $k$ of the balls.  

To find a small separator of a $k$-ply neighborhood system, both a randomized linear time algorithm~\cite{separator-sphere-packing_MillerEtal-jacm1997} and a deterministic linear time algorithm~\cite{separator-k-ply-neighborhood_EppsteinMT-SCG93} are known. Here we use the randomized linear time algorithm by Miller, Teng, Thurston and Vavasis~\cite{separator-sphere-packing_MillerEtal-jacm1997}, which is simpler than the deterministic algorithm by by Eppstein, Miller and Teng~\cite{separator-k-ply-neighborhood_EppsteinMT-SCG93}. Let $\Lambda=\{B_1,\ldots,B_n\}$ be a $k$-ply neighborhood system in $\R^d$. Let $S$ be a sphere in $\R^d$. Let $\Lambda_I(S)$ denote the set of balls of $\Lambda$ in the exterior of $S$, let $\Lambda_E(S)$ denote the set of balls of $\Lambda$ in the exterior of $S$, and let $\Lambda_O(S)$ denote the set of balls of $\Lambda$ cut by $S$. 
\begin{theorem}~\cite{separator-sphere-packing_MillerEtal-jacm1997}\label{thm:sphere-separator-k-ply-randomized}
    Suppose $\Lambda=\{B_1,\ldots,B_n\}$ is a $k$-ply system of $n$ balls in $\R^d$. We can compute a sphere $S$ such that $|\Lambda_O(S)|=O(k^{1/d}n^{1-1/d})$\footnote{The big-$O$ notation hides constant depending on $d$.} and $|\Lambda_I(S)|,|\Lambda_E(S)|\leq \delta\cdot n$ for any $(d+1)/(d+2)<\delta<1$ with probability at least $1/2$. The running time of the algorithm is $c(\epsilon,d)+O(dn)$ where $\epsilon=\delta-(d+1)/(d+2)$, $c(\epsilon,d)$ is a constant depending only on $\epsilon$ and $d$.
\end{theorem}

Let $P=\{p_1,\ldots,p_n\}$ be the centers of $\Lambda=\{B_1,\ldots,B_n\}$. The separator algorithm works as follows.
\begin{enumerate}
    \item Compute a \emph{conformal mapping $\Phi$} from $\R^d$ to a unit sphere $U_d$ in $\R^{d+1}$ such that (i) the center of $U_d$ is a \emph{$\delta$-centerpoint} of $\Phi(P)=\{\Phi(p_1),\ldots,\Phi(p_n)\}$ where $(d+1)/(d+2)<\delta<1$, (ii) each sphere in $\R^d$ is mapped to a $(d-1)$-sphere on $U_d$, (iii) ball $B_i$, $1\leq i\leq n$, is mapped to a cap $\Phi(B_i)$ on $U_d$. 
    \item Choose a random great circle $GS$ on $U_d$.
    \item Transform the great circle $GS$ to a sphere $S$ in $\R^d$ by inversing the conformal mapping $\Phi$. 
    \item Return $S$, $\Lambda_I(S)$, $\Lambda_E(S)$ and $\Lambda_O(S)$. 
\end{enumerate}

A $\delta$-centerpoint of a set $P$ of points in $\R^d$ is a point $\mathbf{c}\in \R^d$ such that every hyperplane through $\mathbf{c}$ $\delta$-splits $P$. A $\delta$-centerpoint where $d/(d+1)<\delta <1$ can be computed in deterministic linear time~\cite{geometric-separator-thesis-teng1991}. A great circle on $U_d$ is the intersection of $U_d$ with a hyperplane passing through the center of $U_d$. The conformal mapping $\Phi$ has the property that the pre-image of any great circle $\mathcal{G}$ on $U_d$ is a $(d-1)$-sphere in $\R^d$, the interior and the exterior of the pre-image sphere are mapped to the two hemispheres of $U_d$ defined by $\mathcal{G}$, respectively. It can be inferred from the proof of Theorem~\ref{thm:sphere-separator-k-ply-randomized} that a random sphere in $\R^d$ cuts (intersects) $O(k^{1/d}n^{1-1/d})$ balls in $\Lambda$ in expectation. 

Until now, we have verified the intersection graph of a $k$-ply neighborhood system in $\R^d$ satisfies criteria (I) and (III) with $c_1=k^{1/d}$, $c_2=d$ and $c_4=k^{1/d}$. The closed surface $\xi$ used to find the separator is a sphere in $\R^d$. The geometrical objects cut by $\xi$ are balls of the neighborhood system.

From the definition of $k$-ply neighborbood system, its intersection graph is subgraph-closed. Criterion (IV) is thus satisfied. It remains to verify criterion (II).

 

\subsection{Sparsity of the contracted graphs}\label{ssec:sparsity-contracted-k-ply}
As a preliminary, we give the definition of the \emph{degeneracy} of a graph. 
\begin{definition}\label{def:degeneracy}
    The degeneracy of a graph $G$ is the maximum, over all subgraphs of $G$, of the minimum vertex degree of the subgraph.
\end{definition}
Let $G$ be the intersection graph of the $k$-ply neighborhood system in $\R^d$. Let $v$ be a vertex in $G$ and let $ball(v)$ be the ball associated with $v$. 

\begin{lemma}\label{lem:k-ply-neighborhood-degeneracy-k}
    The degeneracy of $G$ is at most $3^dk$.
\end{lemma}
\begin{proof}
    It is proved in Lemma~3.3.2 of~\cite{separator-sphere-packing_MillerEtal-jacm1997} that for any vertex $v$ of $G$, $ball(v)$ is intersected by at most $3^d k$ other balls of greater radius in the $k$-ply neighborhood system. 

    Order the balls in the $k$-ply neighborhood system in increasing order of their radii. In this ordering, each ball is intersected by at most $3^d k$ balls appearing later in the ordering. Thus by this ordering each vertex in $G$ has at most $3^d k$ later neighbors. Therefore any subgraph of $G$ has a vertex of degree at most $3^d k$. The degeneracy of $G$ is at most $3^d k$.
\end{proof}

Since $G$ has degeneracy at most $3^d k$, the arboricity and the thickness of $G$ are at most $3^d k$.
\begin{corollary}\label{col:thickness-k-ply-neighborhood}
    The thickness of $G$ is at most $3^dk$.
\end{corollary}

Let $G_i$, $1\leq i\leq I+1$, be the contracted graphs. Using the arguments in Section~\ref{ssec:sparsity-contracted-lanky}, it follows from Corollary~\ref{col:thickness-k-ply-neighborhood} that $G_i$ has $O(3^dk n_i)$ edges where $n_i$ is the number of vertices in $G_i$. Similarly,
\begin{corollary}\label{col:sparsity-contracted-subgraph-k-ply}
    Let $H$ be any subgraph of $G_i$ with $K$ vertices. $H$ has $O(3^dk\cdot K )$ edges.
\end{corollary}

Corollary~\ref{col:sparsity-contracted-subgraph-k-ply} justifies that the intersection graph of a $k$-ply neighborhood system satisfies criterion (II) with $c_3=3^dk$. The intersection graph has $O(3^dkn)$ edges so $c_0=3^dk$. Plugging the constants into Theorem~\ref{thm:recursive-division-main}, we have
\begin{theorem}\label{thm:main-sssp-k-ply}
    Let $G$ be the edge-weighted intersection graph of a $k$-ply neighborhood system in $\R^d$. A $(\bar{r},f)$-recursive division of $G$ satisfying Inequalities~\eqref{eqn:tricky-condition} can be computed in $O(d^29^dk^2 n)$ time. 
\end{theorem}
As a corollary,
\begin{corollary}
    Let $G$ be the edge-weighted intersection graph of a $k$-ply neighborhood system in $\R^d$. Single-source shortest path in $G$ can be computed in $O(d^29^dk^2 n)$ time.
\end{corollary}

\section{SSSP for the intersection graph of \texorpdfstring{$\kappa$}{kap}-thick cubical neighborhood system}\label{ssec:sssp-kappa-thick-cubical}
In this section we consider single-source shortest path in the edge-weighted intersection graph of a $\kappa$-thick cubical neighborhood system in $\R^d$. Smith and Wormald~\cite{geometric-separator-theorems-focs98} showed generalized versions of geometric separator theorems. We define a $\kappa$-thick cubical neighborhood system in $\R^d$ to be a collection of $n$ iso-oriented cubes in $\R^d$ such that no point in the space is covered by more than $\kappa$ of the cubes.  

Let $\Lambda$ be a $\kappa$-thick cubical neighborhood system in $\R^d$. Smith and Wormald proved that $\Lambda$ admits a sublinear separator.
\begin{theorem}[\cite{geometric-separator-theorems-focs98}]~\label{thm:separator-kappa-thick}
Let $\Lambda$ be a $\kappa$-thick cubical neighborhood system in $\R^d$. We can compute an iso-oriented $d$-rectangle with at most $(\frac{2}{3}+\epsilon) n$ cubes in $\Lambda$ entirely inside it, at most $(\frac{2}{3}+\epsilon) n$ cubes in $\Lambda$ entirely outside it, and the number of cubes partly inside and partly outside it is $O(\kappa^{1/d} n^{1-1/d})$\footnote{The big-$O$ notation hides constant depending on $d$}. Such a $d$-rectangle can be computed in $O((d/\epsilon)^{O(d)} + dn)$ expected time. 
\end{theorem}

The separator algorithm~\cite{geometric-separator-theorems-focs98} proceeds by finding a separating $d$-annulus. A separating $d$-annulus is the set difference of two concentric $d$-rectangles of constant aspect ratio such that at least a constant fraction of the $d$-cubes lie inside the inner $d$-rectangle, at least a constant fraction of the $d$-cubes lie outside the outer $d$-rectangle. A randomized approach to finding the separating $d$-annulus is to pick a random subset $q$ of the $n$ cubes and find a separating $d$-annulus of the subset by brute force in $O(|q|^{d+1})$ time. If $|q|=\Omega(\epsilon^{-2}d^3 \log (d/\epsilon))$, the resulting separating annulus succeeds with constant probability. The randomized approach easily generalizes to other convex objects with bounded aspect ratio. A random $d$-rectangle in-between the inner $d$-rectangle and the outer $d$-rectangle of the separating annulus cuts $O(\kappa^{1/d}n^{1-1/d})$ $d$-cubes in $\Lambda$ in expectation. Thus (the surface of) the random $d$-rectangle is the desired separator. It can be inferred from the proof of Theorem 7 in~\cite{geometric-separator-theorems-focs98} that a random $d$-rectangle cuts $O(\kappa^{1/d}n^{1-1/d})$ $d$-cubes in $\Lambda$ in expectation. 

Thus, we have verified that the intersection graph of a $\kappa$-thick cubical neighborhood system in $\R^d$ satisfied criteria (I) and (III) with $c_1=\kappa^{1/d}$, $c_2=d$ and $c_4=\kappa^{1/d}$. The closed surface $\xi$ used to find the separator is the surface of a $d$-rectangle with bounded aspect ratio. The geometrical objects cut by $\xi$ are cubes in $\Lambda$. 

From the definition, the intersection graph of a $\kappa$-thick cubical neighborhood system is subgraph-closed. It remains to verify criterion (II).

Let $G$ be the edge-weighted intersection graph of a $\kappa$-thick cubical neighborhood system in $\R^d$. Let $v$ be a vertex in $G$ and let $cube(v)$ be the cube associated with $v$. We can prove that $cube(v)$ can be intersected by at most $3^d\kappa$ other cubes of greater size in the neighborhood system, using arguments similar to~\cite{separator-sphere-packing_MillerEtal-jacm1997}. Thus
\begin{lemma}
    The degeneracy of $G$ is at most $3^d\kappa$.
\end{lemma}

Since $G$ has degeneracy at most $3^d \kappa$, the thickness of $G$ is at most $3^d \kappa$. Let $G_i$, $1\leq i\leq I+1$, be the contracted graphs of $G$. Using the arguments in Section~\ref{ssec:sparsity-contracted-lanky}, it follows from that $G_i$ has $O(3^d\kappa n_i)$ edges where $n_i$ is the number of vertices in $G_i$. Similarly,
\begin{corollary}\label{col:sparsity-contracted-subgraph-kappa-thick}
    Let $H$ be any subgraph of $G_i$ with $K$ vertices. $H$ has $O(3^d\kappa\cdot K )$ edges.
\end{corollary}

Corollary~\ref{col:sparsity-contracted-subgraph-kappa-thick} justifies that the intersection graph of a $\kappa$-thick cubical neighborhood system satisfies criterion (II) with $c_3=3^d\kappa$. The intersection graph has $O(3^d\kappa n)$ edges so $c_0=3^d\kappa$. Plugging the constants into Theorem~\ref{thm:recursive-division-main} and Corollary~\ref{col:sssp-main}, we have
\begin{theorem}\label{thm:main-sssp-kappa-thick}
    Let $G$ be the edge-weighted intersection graph of a $\kappa$-thick cubical neighborhood system in $\R^d$. A $(\bar{r},f)$-recursive division of $G$ satisfying Inequalities~\eqref{eqn:tricky-condition} can be computed in $O(d^29^d\kappa^2 n)$ time. Single-source shortest path in $G$ can be computed in $O(d^29^d\kappa^2 n)$ time.
\end{theorem}

\bibliographystyle{plainurl}
\bibliography{reference.bib}

\end{document}